\documentclass[journal]{IEEEtran}

\usepackage{amsmath}  %cmex10 ensures that only type 1 fonts will be used
\usepackage{amsfonts,amssymb}
\usepackage{graphicx,color}   %color for pdf_t generated by xfig

\usepackage{ifxetex}
\ifxetex
\usepackage{xunicode}% unicode character macros
\else
\usepackage[utf8]{inputenc}
\usepackage[T1]{fontenc}
\usepackage{microtype}
\fi

\usepackage{cite}  %combined consecutive references
\newtheorem{theorem}{Theorem}
\newtheorem{remark}{Remark}
\newtheorem{lemma}{Lemma}

\newcounter{mytempeqncnt}

\begin{document}

%---------- Title ----------
\title{The Half-Duplex AWGN Single-Relay Channel: Full Decoding or Partial Decoding?}
\author{Lawrence Ong, Sarah J. Johnson, and Christopher M. Kellett}

\maketitle

\begin{abstract}
This paper compares the partial-decode-forward and the complete-decode-forward coding strategies for the half-duplex Gaussian single-relay channel. We analytically show that the rate achievable by partial-decode-forward outperforms that of the more straightforward complete-decode-forward by at most 12.5\%. Furthermore, in the following asymptotic cases, the gap between the partial-decode-forward and the complete-decode-forward rates diminishes: (i) when the relay is close to the source, (ii) when the relay is close to the destination, and (iii) when the SNR is low. In addition, when the SNR increases, this gap, when normalized to the complete-decode-forward rate, also diminishes.  Consequently, significant performance improvements are not achieved by optimizing the fraction of data the relay should decode and forward, over simply decoding the entire source message.
\end{abstract}


\IEEEpeerreviewmaketitle

\section{Introduction}

Wireless relay networks are ubiquitous from classical satellite communications to the increasing use of wireless ad hoc networks where mobiles, tablets, or laptops have the ability to act as wireless relays for other similar devices.  A common coding strategy deployed in relay networks is for the relays to decode the source message, and then forward a function of the message (e.g., some parity bits) to the destination to facilitate decoding. In this context, one can use a {\it partial-decode-forward (PDF)} scheme where the relays decode only a fraction of the source message and forward a function of the decoded part, or a {\it complete-decode-forward (CDF)} scheme where the relays fully decode the source message and forward a function of it. Another consideration in designing wireless relay networks is whether to use a full-duplex (communication possible in both directions simultaneously without interference) or a half-duplex (communication possible in only one direction at a time) scheme. The choice of full- versus half-duplex, as well as the choice of CDF or PDF, impacts on the achievable communication rate. Achievable rates have been demonstrated for the full-duplex relay channel using CDF \cite{covergamal79,kramergastpar04} and PDF \cite{covergamal79},\cite[Thm.\ 16.3]{elgamalkim2001}, and for the half-duplex relay channel using CDF \cite{ongwangmotani08allerton} and PDF \cite{hostmadsen02,hostmadsenzhang05}.

PDF includes CDF and direct transmission as special cases---the relay decodes all source messages in the former and decodes nothing in the latter.
For the additive white Gaussian noise (AWGN) full-duplex relay channel, it has been shown~\cite{elgamalmehseni06} that (i) when the source-destination link is better than the source-relay link, the maximum PDF rate can be achieved by not using the relay at all, i.e., direct transmission; and (ii) when the source-relay link is better than the source-destination link, the maximum PDF rate can be achieved by having the relay decoding all source messages, i.e., CDF.

While it is not desirable to have the relay decoding only part of the messages in the AWGN full-duplex relay channel, PDF can achieve rates strictly higher than direct transmission and CDF in the half-duplex counterpart.

\subsection{Main Results}

In this paper, we focus on the half-duplex relay channel, and compare the rates achievable by existing coding schemes.
For a given half-duplex relay channel, let $R_\text{CDF}$ and $R_\text{PDF}$ be the maximum rates achievable (optimized over channel codes under respective coding schemes) using the complete-decode-forward and the partial-decode-forward schemes respectively.
For the trivial case where the source-destination link is better than the source-relay link, $R_\text{PDF}$  can be attained by direct transmission (similar to the full-duplex case). Otherwise, we have the following:
\begin{theorem} \label{theorem:main}
Consider the AWGN half-duplex single-relay channel. If the source-relay link is better than the source-destination link, then %that
\begin{equation}
R_\text{CDF} \leq R_\text{PDF} \leq \frac{9}{8} R_\text{CDF}.
\end{equation}
\end{theorem}
This means regardless of channel parameters, partial decoding can achieve at most 12.5\% rate gain compared to full decoding.

We further show that in the following three asymptotic cases, PDF and CDF actually achieve the same rates: (i) the relay is close to the source, (ii) the relay is close to the destination, and (iii) the signal-to-noise ratio (SNR) is low. We also show that for the case of high SNR, CDF achieves rates within a constant bit gap (which is a function of the channel gains) of that achievable by PDF. This gap, when normalized to the CDF rate, diminishes as the SNR increases. Theoretically, this work gives an upper bound to the performance gain of PDF over CDF. Practically, it shows that without a significant percentage loss in transmission rate, one could choose CDF over PDF, as the former is potentially less complex to implement, and is certainly less complex to design.

\section{Background}

\subsection{Channel Model}
Fig.~\ref{fig:src} denotes the single-relay channel consisting of the source (denoted by node 0), the relay (node 1), and the destination (node 2).
We consider the AWGN channel where (i) the received signal at the relay is given by $Y_1 = h_{01} X_0 + Z_1$, and (ii) the received signal at the destination is given by $Y_2 = h_{02} X_0 + h_{12} X_1 + Z_2$. Here, $X_i$ is the transmitted signal of node $i$, $Z_j$ is the noise at node $j$, and $h_{ij}$ is the channel gain from node $i$ to node $j$. Each $Z_j$ is an independent zero-mean Gaussian random variable with variance $E[Z_j^2] = N_j$. The channel gains are fixed and are made known a priori to all the nodes. 

\begin{figure}[t]
\centering
%\scalebox{0.8}{\input{../pictures/pdf-08.pspdftex}}
%\input{../pictures/pdf-03.pspdftex}
\includegraphics[width=\linewidth]{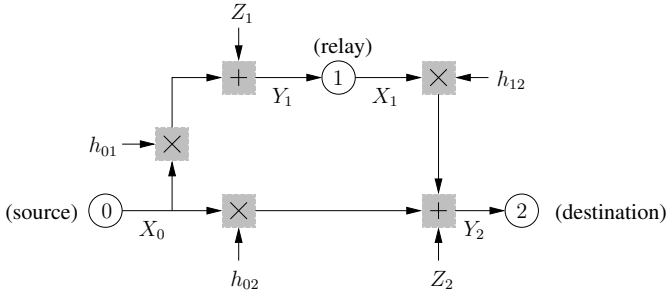}
\caption{The AWGN single-relay channel}
\label{fig:src}
\end{figure}

\subsection{Rate Definition}

Consider a block of $n$ channel uses, and let $x_{it}$ (and $y_{it}$) denote the transmitted (and received) symbol of node $i$ at the $t$-th channel use. 
A rate-$R$ block code comprises (i) a message set $\mathcal{W} = \{1,2,\dotsc,2^{nR}\}$, (ii) an encoding function for the source $\boldsymbol{x}_0 = f_0(w)$ for each message $w \in \mathcal{W}$, (iii) encoding functions for the relay $x_{1t} = f_{1t}(y_{11}, y_{12}, \dotsc, y_{1t-1})$ for all $t \in \{1,2,\dotsc, n\}$, and (iv) a decoding function for the destination $\hat{w} = f_2(\boldsymbol{y}_2)$. Here, we have used bold-faced symbols to represent vectors of length $n$, e.g., $\boldsymbol{x}_0 = (x_{01}, x_{02}, \dotsc, x_{0n})$.
We assume that the message is uniformly distributed in $\mathcal{W}$.
The rate $R$ is achievable if there exists a sequence of rate-$R$ block codes such that $\Pr \{\hat{W} \neq W\} \rightarrow 0$ as $n \rightarrow \infty$.

\subsection{Half-Duplex and Power Constraints}

For the half-duplex channel, the relay can only transmit or receive (i.e., listen), but not do both, at any time. We set $Y_{1t} = 0$ if the relay transmits at the $t$-th channel use, and $X_{1t}=0$ otherwise, i.e., if the relay listens.
 Note that the definition of block code above is applicable for both full-duplex and half-duplex relay channels with these extra transmit/listen constraints. Without loss of optimality, the source always transmits, and the destination always listens.

We consider a {\em fixed-slot} structure where the transmit/listen mode of the relay for each channel use is fixed prior to the transmissions and is known to all nodes~\cite[p.\ 348]{kramermaric06}. 
Consider $n$ channel uses during which the relay listens in $\alpha n$ channel uses and transmits in $(1-\alpha)n$ channel uses. We assume the following {\em per-symbol} power constraints~\cite[p.\ 304]{kramermaric06}: $E[X_{it}^2] \leq P_i$ for both $i \in \{0,1\}$ and for all $t \in \{1,2,\dotsc,n\}$, where the expectation operation $E[\cdot]$ is taken over the message $W$ and the channel noise $Z_1$. As a result, the source and the relay cannot optimize their transmit power for variations in the relay's transmit/listen mode.\footnote{In addition to modeling instantaneous power constraints, this assumption also simplifies the analyses in this paper as the instantaneous transmit power of both the source and the relay is not a function of $\alpha$.} We denote the SNR from node $i$ to node $j$ by $\lambda_{ij} \triangleq \frac{ h_{ij}^2 P_i}{ N_j}$. We assume that $\lambda_{ij} >0$ for all $i$ and $j$.

\subsection{Achievable Rates}

We now summarize the encoding and decoding schemes of PDF for the half-duplex relay channel~\cite{hostmadsen02,hostmadsenzhang05}. Consider $B$ blocks, each of $n$ channel uses, and $(B-1)$ source messages $\{W^{<i>}\}_{i=1}^{B-1}$, where each $W^{<i>} \in \{1,\dotsc, 2^{nR}\}$. The source splits each message $W^{<i>}$ into two independent parts $W^{<i>} = (U^{<i>},V^{<i>})$, where $U^{<i>} \in \{1,\dotsc, 2^{nR_U}\}$ and $V^{<i>} \in \{1,\dotsc, 2^{nR_V}\}$ so that $R = R_U + R_V$. In block $b \in \{1,\dotsc, B-1\}$, the source transmits $(U^{<b-1>},U^{<b>}, V^{<b>})$. The relay decodes $U^{<b>}$ in block $b$, and transmits it in block $(b+1)$. The rate of this code over the entire $B$ transmission blocks is $\frac{(B-1)nR}{nB}$ bits/channel use. If the destination can decode all $\{W^{<i>}\}_{i=1}^{B-1}$ reliably (i.e., with diminishing error probability as $n$ increases), then by choosing a sufficiently large $B$,  any rate below and arbitrarily close to $R$ is achievable.

\begin{figure}[t]
\centering
%\resizebox{!}{!}{\input{../pictures/pdf-04.pspdftex}}
%\scalebox{0.8}{\input{../pictures/pdf-07.pspdftex}}
\includegraphics[width=\linewidth]{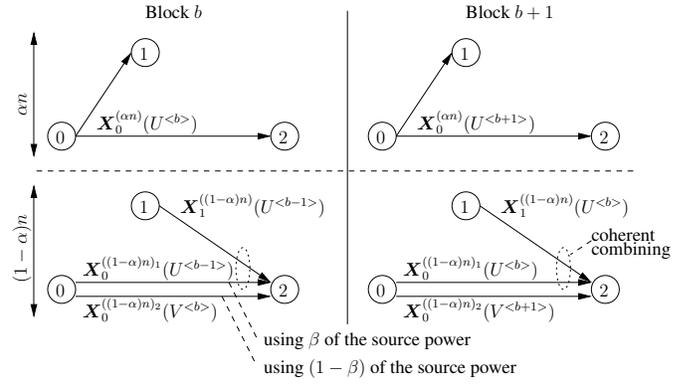}
\caption{PDF on the half-duplex relay channel, where in $\alpha$ of the time the relay listens, and in $(1-\alpha)$ of the time the relay transmits}
\label{fig:pdf-half}
\end{figure}

As the encoding and decoding operations repeat themselves over the blocks, we focus on encoding in block $b$ and decoding over blocks $b$ and $(b+1)$, depicted in Fig.~\ref{fig:pdf-half}. During the first $\alpha n$ channel uses of block $b$ (when the relay listens)\footnote{As an example, we assume that the relay listens in the first $(\alpha n)$ channel uses and transmits in the remaining channel uses of each block. In general, the channel uses in which the relay listens or transmits can be arbitrarily chosen in each block without affecting the achievable rates as long as $\alpha$ is fixed.}, the source transmits $\boldsymbol{X}_0^{(\alpha n)}(U^{<b>})$. Here, we have used the superscript $(\alpha n)$ to denote the length of the sub-codeword. Assume that the relay has decoded $U^{<b-1>}$ in the previous block, i.e., block $(b-1)$. During the remaining $(1-\alpha)n$ channel uses of block $b$,  the relay transmits an independently generated codeword $\boldsymbol{X}_1^{((1-\alpha)n)} (U^{<b-1>})$, while the source splits its transmit power into two parts: (i) a portion of $\beta$ of the power is used to transmit the same codeword as the relay, i.e., $\boldsymbol{X}_0^{((1-\alpha)n)_1}(U^{<b-1>}) = \sqrt{\frac{\beta P_0}{P_1}}\boldsymbol{X}_1 ^{((1-\alpha)n)} (U^{<b-1>})$, so that this portion of the source's signal and the relay's signal add coherently at the destination\footnote{The source transmits the message $U^{<b>}$ twice: $\boldsymbol{X}_0^{(\alpha n)}(U^{<b>})$ in block $b$ and $\boldsymbol{X}_0^{((1-\alpha)n)_1}(U^{<b>})$ in block $(b+1)$. Note that these two codewords are independently generated.}; and (ii) the remaining portion of $(1-\beta)$ of the power is dedicated to a different codeword that is to be decoded by the destination, i.e., $\boldsymbol{X}_0^{((1-\alpha)n)_2}(V^{<b>})$. In brief, the source transmits $\boldsymbol{X}_0 = [ \boldsymbol{X}_0^{(\alpha n)}(U^{<b>}), \boldsymbol{X}_0^{((1-\alpha)n)_1}(U^{<b-1>})  + \boldsymbol{X}_0^{((1-\alpha)n)_2}(V^{<b>})]$ in block $b$.

The above encoding scheme requires that the relay decodes $U^{<b>}$  at the end of block $b$;  the relay can reliably do so  if $R_U \leq \frac{\alpha}{2} \log ( 1 + \lambda_{01})$.
The decoding of the message $W^{<b>}$ at the destination is performed over blocks $b$ and $(b+1)$. Assume that the destination has decoded messages $\{W^{<i>}\}_{i=1}^{b-1}$. Using the last $(1-\alpha)n$ channel uses of block $b$, the destination can decode $V^{<b>}$ if $R_V \leq \frac{1-\alpha}{2} \log( 1 + (1-\beta)\lambda_{02})$. Using the first $\alpha n$ channel uses of block $b$ and the last $(1-\alpha)n$ channel uses of block $(b+1)$, the destination can decode $U^{<b>}$ if $R_U \leq \frac{\alpha}{2} \log (1 + \lambda_{02}) + \frac{1-\alpha}{2} \log \left( 1+  \frac{[ \sqrt{\lambda_{12}} + \sqrt{\beta\lambda_{02}}]^2}{ (1-\beta) \lambda_{02} + 1} \right)$. Here, $[ \sqrt{\lambda_{12}} + \sqrt{\beta \lambda_{02}}]^2 $ is power of the coherently combined signals from the relay and the source, $(1-\beta) \lambda_{02}$ is the power of signals carrying $V^{<b+1>}$ which appears as noise when decoding $U^{<b>}$, and the channel noise power $E[Z_2^2]$ is one. Combining these rate constraints, PDF~\cite{hostmadsen02,hostmadsenzhang05}  achieves rates up to 

\begin{align}
&R_\text{PDF} = \max_{0 \leq \alpha, \beta \leq 1} \min \Big\{ \nonumber \\ 
& \frac{\alpha}{2} \log ( 1  + \lambda_{01} ) + \frac{1-\alpha}{2} \log ( 1 + (1-\beta)\lambda_{02}), \nonumber \\ 
&\frac{\alpha}{2} \log ( 1 + \lambda_{02} ) + \frac{1-\alpha}{2} \log ( 1 + \lambda_{02} + \lambda_{12} + 2 \sqrt{\beta \lambda_{02} \lambda_{12} }) \Big\}. \label{eq:pdf-rate}
\end{align}

Setting $\beta=1$ (which means $V^{<i>} = \varnothing$; i.e., the relay decodes all source messages), we have CDF~\cite{ongwangmotani08allerton}, which achieves rates up to  

\begin{align}
&R_\text{CDF} = \max_{0 \leq \alpha \leq 1} \min \Big\{ \frac{\alpha}{2} \log ( 1  + \lambda_{01} ), \nonumber \\
& \frac{\alpha}{2} \log ( 1 + \lambda_{02} ) + \frac{1-\alpha}{2} \log ( 1 + [\sqrt{\lambda_{02}} + \sqrt{\lambda_{12}}]^2) \Big\}. \label{eq:cdf-rate}
\end{align}
Clearly, $R_\text{CDF} \leq R_\text{PDF}$. Furthermore, setting $\alpha = 0$ and $\beta = 0$ (which means $U^{<i>} = \varnothing$; i.e., the relay decodes nothing), we have direct transmission.

When $\lambda_{01} < \lambda_{02}$ (the source-destination link is better than the source-relay link), $R_\text{PDF} = \frac{1}{2} \log ( 1 + \lambda_{02})$, which is attained only at $\alpha = \beta = 0$, i.e., with direct transmission.

When $\lambda_{01} = \lambda_{02}$ (the source-destination link is as good as the source-relay link), we again have $R_\text{PDF} = \frac{1}{2} \log ( 1 + \lambda_{02})$, which is attained by $\beta = 0$ and any $\alpha \in [0,1]$. Although the relay is used when $\alpha> 0$, PDF still achieves the same rate as that of direct transmission, $\alpha = \beta =0$.

For the rest of this paper, we will consider the non-trivial case of $\lambda_{01} > \lambda_{02}$ (when the source-relay link is better than the source-destination link).
We will show later that $R_\text{PDF}$ can be strictly higher than $R_\text{CDF}$. Furthermore, it is easy to see that, as all $\lambda_{ij} > 0$, both PDF and CDF strictly outperform direct transmission.

We define the gap between the PDF and the CDF rates by

\begin{equation}
G \triangleq R_\text{PDF} - R_\text{CDF},
\end{equation}
and the normalized (to $R_\text{CDF}$) gap by

\begin{equation}
\bar{G} \triangleq \frac{ R_\text{PDF} - R_\text{CDF} }{R_\text{CDF} }. \label{eq:normalized-gap}
\end{equation}
%For the rest of this paper, we will use $R_\text{CDF}$ as the normalization term. With this, $\bar{G}$ is the percentage gain in rate one can expect by doing partial decoding (i.e., by optimizing how much of the source message the relay should decode) as compared to the more straight-forward complete decoding.

\section{CDF vs. PDF for $\lambda_{01} > \lambda_{02}$}

First, we have the following lemma:

\begin{lemma}\label{lemma:cdf}
Let $\alpha^*$ be the $\alpha$ that attains $R_\text{CDF}$. If $\lambda_{01} > \lambda_{02}$,  it follows that $\frac{\alpha^*}{2} \log ( 1  + \lambda_{01} ) = \frac{\alpha^*}{2} \log ( 1 + \lambda_{02} ) + \frac{1-\alpha^*}{2} \log ( 1 + [\sqrt{\lambda_{02}} + \sqrt{\lambda_{12}}]^2)$.
\end{lemma}

\begin{proof}[Proof of Lemma~\ref{lemma:cdf}]
Let $a(\alpha) = \frac{\alpha}{2} \log ( 1  + \lambda_{01} )$ and $b(\alpha) = \frac{\alpha}{2} \log ( 1 + \lambda_{02} ) + \frac{1-\alpha}{2} \log ( 1 +[ \sqrt{\lambda_{02}} + \sqrt{\lambda_{12}}]^2)$. The function $a(\alpha)$ is a continuous and increasing in $\alpha$ with $a(0) = 0$ and $a(1) > 0$; $b(\alpha)$ is a continuous and decreasing function of $\alpha$  with $b(0) > 0$ and $b(1) = \frac{1}{2} \log ( 1 + \lambda_{02} ) < a(1)$. So the solution to $\max_{0 \leq \alpha \leq 1} \min \{ a(\alpha), b(\alpha) \}$ is at the point where $a(\alpha)$ and $b(\alpha)$ intersect.
\end{proof}

From the above lemma, we can show that
\begin{lemma}\label{lemma:cdf-rate}
If $\lambda_{01} > \lambda_{02}$, then
\begin{multline}
R_\text{CDF} = \\  \frac{ \frac{1}{2}\log ( 1 + [\sqrt{\lambda_{02}} + \sqrt{\lambda_{12}}]^2) \log ( 1  + \lambda_{01} ) } { \log ( 1 + [\sqrt{\lambda_{02}} + \sqrt{\lambda_{12}}]^2) + \log ( 1 + \lambda_{01}) - \log (1 + \lambda_{02} )}.
\end{multline}
\end{lemma}

\begin{proof}[Proof of Lemma~\ref{lemma:cdf-rate}]
Substituting $\alpha^*$ from Lemma~\ref{lemma:cdf} into \eqref{eq:cdf-rate} gives Lemma~\ref{lemma:cdf-rate}.
\end{proof}

Now, by selecting the $\beta \in [0,1]$ that maximizes the two terms on the right-hand side of \eqref{eq:pdf-rate} (i.e., $\beta=0$ and $\beta=1$ respectively), we can upper bound $R_\text{PDF}$ as

%\begin{subequations}
\begin{align}
 R_\text{PDF}
&\leq \max_{0 \leq \alpha \leq 1} \min \Big\{ \frac{\alpha}{2} \log ( 1  + \lambda_{01} ) + \frac{1-\alpha}{2}
\log ( 1 + \lambda_{02}),\nonumber\\
& \quad \frac{\alpha}{2} \log ( 1 + \lambda_{02} ) + \frac{1-\alpha}{2}
\log ( 1 + [\sqrt{\lambda_{02}} + \sqrt{\lambda_{12}}]^2 ) \Big\}\nonumber \\
& \triangleq R_\text{PDF}^\text{UB}, \label{eq:pdf-ub} 
\end{align}
%\end{subequations}
where $R_\text{PDF}^\text{UB}$ is an upper bound to $R_\text{PDF}$.

Using the same arguments in the proof of Lemma~\ref{lemma:cdf}, we obtain the following:
\begin{lemma}\label{lemma:pdf-ub-a}
Let $\alpha^\dagger$ be the $\alpha$ that attains
$R_\text{PDF}^\text{UB}$. If $\lambda_{01} > \lambda_{02}$ then $\frac{\alpha^\dagger}{2} \log ( 1  + \lambda_{01} ) + \frac{1-\alpha^\dagger}{2}
\log ( 1 + \lambda_{02}) = \frac{\alpha^\dagger}{2} \log ( 1 + \lambda_{02} ) + \frac{1-\alpha^\dagger}{2}
\log ( 1 + [\sqrt{\lambda_{02}} + \sqrt{\lambda_{12}}]^2) $.
\end{lemma}

Hence, the upper bound to $R_\text{PDF}$ is given as follows:
\begin{lemma}\label{lemma:pdf-ub-b}
If $\lambda_{01} > \lambda_{02}$, then

\begin{align}
&R_\text{PDF}^\text{UB}  = \nonumber \\ 
&\frac{\frac{1}{2} \log ( 1 + [\sqrt{\lambda_{02}} + \sqrt{\lambda_{12}}]^2) \log ( 1 + \lambda_{01}) }{ \log ( 1 + [\sqrt{\lambda_{02}} + \sqrt{\lambda_{12}}]^2) + \log ( 1 + \lambda_{01}) - 2 \log ( 1 + \lambda_{02}) } \nonumber \\
&- \frac{ \frac{1}{2} [\log ( 1 + \lambda_{02})]^2 }{ \log ( 1 + [\sqrt{\lambda_{02}} + \sqrt{\lambda_{12}}]^2) + \log ( 1 + \lambda_{01})  - 2 \log ( 1 + \lambda_{02}) }.
\end{align}
\end{lemma}

\begin{proof}[Proof of Lemma~\ref{lemma:pdf-ub-b}]
Substituting $\alpha^\dagger$ in Lemma~\ref{lemma:pdf-ub-a} into \eqref{eq:pdf-ub} gives Lemma~\ref{lemma:pdf-ub-b}.
\end{proof}

\subsection{An Upper Bound on the Normalized Gap}
% The previous equation was number five.
 % Account for the double column equations here. 
%\addtocounter{equation}{1}

We now derive the following upper bound on the normalized gap~\eqref{eq:normalized-gap} between the PDF and the CDF rates:
\begin{equation}
\bar{G} \leq 1/8. \label{eq:upper-gap}
\end{equation}

First, we define the following: $w \triangleq \log ( 1 + \lambda_{01} )$, $u \triangleq \log ( 1 + \lambda_{02} )$, $v \triangleq w -2u$, $t \triangleq \frac{u(w-u)}{w}$, and $q \triangleq \log\left(1 + [\sqrt{\lambda_{02}} + \sqrt{\lambda_{12}}]^2\right)$.
Note that while $w$, $u$, $v$, and $t$ are completely determined by the channel gains, the source's transmit power, and the receiver noise levels at the relay and the destination, $q$ is additionally determined by the relay's transmit power (in addition to the channel gains, the source's transmit power, and the noise levels). With the above definitions, we have the normalized gap upper bounded by
\begin{equation}
\bar{G} = \frac{R_\text{PDF} - R_\text{CDF}}{R_\text{CDF}} \leq \frac{R_\text{PDF}^\text{UB} - R_\text{CDF}}{R_\text{CDF}}= \frac{t(q-u)}{q(q+v)} \triangleq \bar{G}^\text{UB}. \label{eq:g-bar-ub}
\end{equation}

To prove \eqref{eq:upper-gap}, we first show the following lemma:
\begin{lemma} \label{lemma:gap}
If $\lambda_{01} > \lambda_{02}$, then
\begin{equation}
\bar{G}^\text{UB} \leq \frac{u(w-u)}{w\left[ w + 2\sqrt{u(w-u)}\right] }. \label{eq:lemma}
\end{equation}
\end{lemma}

\begin{proof}[Proof of Lemma~\ref{lemma:gap}]
We first write $\bar{G}^\text{UB} = t f(q)$ where $f(q) = \frac{(q-u)}{q(q+v)}$. Note that for any fixed channel gains, $P_0$, $N_1$, and $N_2$, the variables $t$, $u$, and $v$ are fixed, but the variable $q$ can take any value by choosing an appropriate $P_1$ subject to the constraint $q  \triangleq  \log\left(1 + [\sqrt{\lambda_{02}} + \sqrt{\lambda_{12}}]^2\right)  > \log ( 1 + \lambda_{02} )  \triangleq u $. This means $q$ can be chosen independent of $t$, $u$, and $v$. Now, for any $t$, $u$, and $v$, we maximize $f(q)$ with respect to $q$. Note that as $\lambda_{01} > \lambda_{02}$, we must have $v > -u$. So, for the region of interest, i.e., where $q > u > -v$, the function $f(q)$ is continuously differentiable. By differentiating $f(q)$ with respect to $q$, we get $\frac{d f(q)}{d q} = \frac{ -q^2 + 2uq + uv } {[q(q+v)]^2}$.
The denominator is always positive, and the numerator is a quadratic function with roots $q_1 = u - \sqrt{u(u+v)}$ and $q_2 = u + \sqrt{u(u+v)}$. Note that $q_1 < u < q_2$ since $u + v > 0$. This means
\begin{equation}
-q^2 + 2uq + uv \begin{cases}
> 0, & \text{ if } u < q < q_2\\
=0, & \text{ if } q = q_2\\
< 0,& \text{ if } q > q_2.
\end{cases}
\end{equation}
So, the maximum of $f(q)$ for $q > u$ occurs at $q_2 = u + \sqrt{u(u+v)}$. Thus,
\begin{equation}
 \bar{G}^\text{UB} \leq t \times \max_{q>u} f(q) = \frac{t\sqrt{u(u+v)}}{[u + \sqrt{u(u+v)}][u + \sqrt{u(u+v)}+v]}.
\end{equation}
Substituting $u+v = w-u$ and $t=u(w-u)/w$ into the above equation, we get \eqref{eq:lemma}.
\end{proof}

With the above lemma, we now prove \eqref{eq:upper-gap}.
First, we define $s \triangleq w/u$, where $s > 1$ because $w > u$. The right-hand side of \eqref{eq:lemma} can be written as $1/h(s)$ where $h(s) = s\left[ 1 + 2 (s-1)^{-\frac{1}{2}} + (s-1)^{-1} \right]$.
This means $\bar{G}^\text{UB} \leq 1/ [ \min_{s>1} h(s)]$. Note that $h(s)$ is continuously differentiable for all $s > 1$, and its first derivative is
%\vspace{-1ex}
\begin{equation}
\frac{dh(s)}{ds} = \frac{(s-1)^2 + (s-2)\sqrt{s-1} -1}{(s-1)^2} 
\begin{cases}
< 0, &\text{ if } 1 < s < 2\\
=0, &\text{ if } s = 2\\
> 0, & \text{ if } s > 2.
\end{cases}
\end{equation}
So, $\min_{s>1} h(s) = h(2) = 8$. This gives \eqref{eq:upper-gap}.

Combining \eqref{eq:upper-gap} and the trivial lower bound $R_\text{CDF} \leq R_\text{PDF}$, we have Theorem~\ref{theorem:main}. $\hfill \blacksquare$

\subsection{A Note on the Trivial Lower Bound to $\bar{G}$}

Suppose that $\lambda_{01} \gg \lambda_{02} + \lambda_{12}$. The optimal $\alpha$ for the max-min operation in \eqref{eq:cdf-rate} is close to zero. Hence, $R_\text{CDF} \approx \frac{1}{2} \log( 1 + [ \sqrt{\lambda_{02}} + \sqrt{\lambda_{12}} ] ^2)$. For \eqref{eq:pdf-rate}, the optimal $\alpha$ is close to zero and the optimal $\beta$ is close to one. Hence, $R_\text{PDF} \approx \frac{1}{2} \log( 1 + [ \sqrt{\lambda_{02}} + \sqrt{\lambda_{12}} ] ^2)$. So the trivial lower bound $R_\text{CDF} \leq R_\text{PDF}$ is (almost) tight when $\lambda_{01} \gg \lambda_{02} + \lambda_{12}$.

\subsection{A Note on the Upper Bound to $\bar{G}$}

Numerical results show that the normalized gap $\bar{G}$ approaches the upper bound of $12.5\%$ (in Theorem~\ref{theorem:main}) as $\lambda_{12}$ increases (for selected $\lambda_{01}$ and $\lambda_{02}$ for each $\lambda_{12}$). For example, setting $\lambda_{12} = 10^5$, $\lambda_{01}=62000$, and $\lambda_{02}=230$, the normalized gap is $\bar{G}= 12.2\%$.

\section{Case Study: The Path Loss Model}

In this section, we consider the free-space path loss model where the channel gain is given by $h_{ij} = d_{ij}^{-1}$, where $d_{ij}$ is the distance between the transmitter and the receiver, and so $\lambda_{ij} = \frac{P_{i}}{d_{ij}^{-2}N_j}$.

\begin{figure*}[!t] 
% ensure that we have normalsize text
\normalsize 
% Store the current equation number.
\setcounter{mytempeqncnt}{\value{equation}} 
% Set the equation number to one less than the one
% desired for the first equation here.
% The value here will have to changed if equations
% are added or removed prior to the place these
% equations are referenced in the main text.
%\setcounter{equation}{5}
\begin{subequations}
  \begin{align}
    G &= R_\text{PDF} - R_\text{CDF} \leq R_\text{PDF}^\text{UB} - R_\text{CDF} \\
& ={\textstyle \frac{ \Big[ \log(1+\lambda_{01})-\log(1+\lambda_{02})\Big]
      \Big[ \log(1 + [\sqrt{\lambda_{02}} + \sqrt{\lambda_{12}}]^2) -
      \log(1+\lambda_{02}) \Big] \log(1+\lambda_{02})}{ 2\Big[\log ( 1 +
      [\sqrt{\lambda_{02}} + \sqrt{\lambda_{12}}]^2) + \log ( 1 +
      \lambda_{01}) - 2 \log ( 1 + \lambda_{02})\Big] \Big[\log ( 1 +
      [\sqrt{\lambda_{02}} + \sqrt{\lambda_{12}}]^2) + \log ( 1 +
      \lambda_{01}) - \log ( 1 + \lambda_{02})\Big]} } \triangleq G^\text{UB}. \label{eq:g-2}
%&= \frac{(w-u)(q-u)u}{2(q+v)(q+w-u)},
  \end{align}
\end{subequations}
% Restore the current equation number.
\setcounter{equation}{\value{mytempeqncnt}}
% IEEE uses as a separator
\hrulefill 
% The spacer can be tweaked to stop underfull vboxes. 
\vspace*{4pt}
\end{figure*}

\subsection{Varying Relay's Position}

We fix coordinates of the source at $(0,0)$ and the destination at $(0,1)$, and compute the actual $\bar{G}$ for different relay coordinates $(x_1,y_1)$.
We assume that % all nodes transmit at the same power, and have the same receiver noise power, i.e., 
$P_0 = P_1 = 100$, and $N_1 = N_2 = 1$.
We are interested in networks where $\lambda_{01} > \lambda_{02}$, which is equivalent to restricting the relay position to $\mathcal{S} = \{(x_1,y_1): \sqrt{x_1^2 + y_1^2} < 1\}$. The result is shown in Fig.~\ref{fig:3d}. While an upper bound derived for $\bar{G}$ in the previous section is $\frac{1}{8} = 12.5\%$ for all channel parameters, for this example of equal power and equal noise, the actual gap is only 6.45\% or less.

\begin{figure}[t]
\centering
%\resizebox{9cm}{!}{\input{../computations/compare-3d.pspdftex}}
%\scalebox{0.9}{\input{../computations/compare-3d-2.pspdftex}}
\includegraphics[width=9cm]{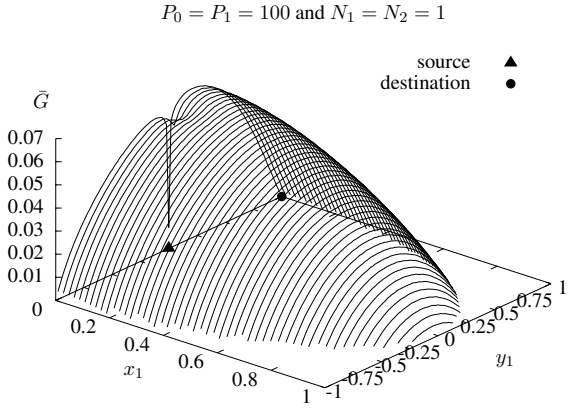}
\caption{A plot showing the normalized difference, $\bar{G} \triangleq \frac{R_\text{PDF}-R_\text{CDF}}{R_\text{CDF}}$, for varying relay positions (only positions at $x_1 \geq 0$ are shown due to symmetry): we have $\max_{(x_1,y_1) \in \mathcal{S}} \bar{G} = 0.0645$ }
\label{fig:3d}
\end{figure}

\subsection{Asymptotic Cases}

\addtocounter{equation}{1}
%In the previous section, we derived an upper bound on $\bar{G}$ for all channel parameters.
Now, we present four cases in which $G$ or $\bar{G}$ tends to zero, meaning that CDF performs as well as PDF. In the following subsections, unless otherwise stated, the transmit power $P_0$ and $P_1$, the inter-node distances $d_{01}$, $d_{02}$, $d_{12}$, and the receiver noise $N_1$, $N_2$ are all positive and finite.
First, we find an upper bound on the gap between the PDF and the CDF rates, given by $G^\text{UB}$ in \eqref{eq:g-2}.

\subsubsection{Relay Close to the Source}
When the relay is close to the source, i.e., $d_{01} \rightarrow 0$, we have $\lambda_{01} \rightarrow \infty$. Since $\lambda_{02}$ and $\lambda_{12}$ are finite, we see that both 
$G^\text{UB} \rightarrow 0$ and $\bar{G}^\text{UB} \rightarrow 0$
as $d_{01} \rightarrow 0$.

\subsubsection{Relay Close to the Destination}
Similarly, when the relay is close to the destination, i.e., $d_{12} \rightarrow 0$, we have $\lambda_{12} \rightarrow \infty$. Since $\lambda_{01}$ and $\lambda_{02}$ are finite, we also see that
$G^\text{UB} \rightarrow 0$ and $\bar{G}^\text{UB} \rightarrow 0$
as $d_{12} \rightarrow 0$.

\subsubsection{High SNR} \label{section:high-snr}
Now, we consider fixed node positions, and increase the transmit power of the nodes. We let $P_0 = k_0P$, $P_1 = k_1P$, $N_1=1$, $N_2=1$ for some constants $k_0$ and $k_1$, and investigate the behavior of $G$ and $\bar{G}$ when $P$ increases. We further define the following, which are independent of $P$: $C_1 \triangleq \frac{\lambda_{01}}{\lambda_{02}} = \frac{d_{01}^{-2}}{d_{02}^{-2}}$ and $C_2 \triangleq \frac{[\sqrt{\lambda_{02}} + \sqrt{\lambda_{12}}]^2}{\lambda_{02}} = \frac{\left[\sqrt{d_{02}^{-2}k_0}+ \sqrt{d_{12}^{-2}k_1}\right]^2}{d_{02}^{-2}k_0}$.

For high SNR, we use the approximation $\lim_{\lambda \rightarrow \infty} \log ( 1 + \lambda) = \log \lambda$, and obtain
\begin{align} 
\lim_{P \rightarrow \infty} G^\text{UB} &= \frac{\log C_1 \log C_2 \log P + \log C_1 \log C_2 \log (d_{02}^{-2}k_0) }{2[ \log (C_1C_2) \log P + \log (C_1C_2) \log (C_2 d_{01}^{-2}k_0) ]}\nonumber \\
&\approx % \frac{ \log C_1 \log C_2}{ 2\log (C_1C_2) } = 
\frac{1}{2} \left( \frac{1}{\log C_1} + \frac{1}{\log C_2}  \right)^{-1}. \label{eq:high-snr}
\end{align}

So, in the high SNR regime where the PDF and the CDF rates are high, the gap between the rates is upper bounded by a constant. This constant, when normalized to the PDF or the CDF rates, approaches zero as the SNR increases, i.e., $\lim_{P \rightarrow \infty} \bar{G}^\text{UB} \rightarrow 0$.  %, as shown below:

\begin{remark} \label{remark:large-g}
From \eqref{eq:high-snr}, we see that $\lim_{P \rightarrow \infty} G^\text{UB}$ increases as $C_1$ and $C_2$ increase. This means with a fixed $\lambda_{02}$, for sufficiently large $\lambda_{01}$ and $\lambda_{12}$, and a much larger $P$, $G^\text{UB}$ can be made arbitrarily large.
\end{remark}

\subsubsection{Low SNR}

Finally, we investigate the case where the SNR tends to zero. We follow the settings (inter-node distances, transmit power, and receiver noise) as in the high-SNR case above, but with $P \rightarrow 0$.

For low SNR, as $P \rightarrow 0$, both $R_\text{PDF}, R_\text{CDF} \rightarrow 0$. So, $\lim_{P \rightarrow 0} G^\text{UB} \rightarrow 0$.

Using the approximation $\lim_{\lambda \rightarrow 0} \log(1+\lambda) = \frac{\lambda}{\ln 2}$, we further see that the normalized gap is % given as follows:
\begin{equation}
\lim_{P \rightarrow 0}\bar{G}^\text{UB}  = C_3 C_4 C_5, \nonumber
\end{equation}
where $C_3 = \left( \frac{d_{01}}{d_{02}} \right)^2$,\\ $C_4 = \frac{(d_{01}^{-2} - d_{02}^{-2})k_0}{ (d_{01}^{-2} - d_{02}^{-2})k_0 + 2d_{02}^{-1}d_{12}^{-1}\sqrt{k_0 k_1}+ d_{12}^{-2}k_1}$, and\\ $C_5 = \frac{ \left(\sqrt{d_{02}^{-2}k_0}+ \sqrt{d_{12}^{-2}k_1} \right)^2- d_{02}^{-2}k_0 }{\left(\sqrt{d_{02}^{-2}k_0}+ \sqrt{d_{12}^{-2}k_1} \right)^2}$ .

As $\lambda_{01} > \lambda_{02}$, we have $d_{01} < d_{02}$. So $C_3,C_4,C_5 < 1$, and it follows that $\lim_{P \rightarrow 0}\bar{G}^\text{UB} < \min_{i \in \{3,4,5\}} C_i$.  Consequently, $\lim_{P \rightarrow 0}\bar{G}^\text{UB} \rightarrow 0$ if any of the following is true: (i) $C_3 \rightarrow 0$, i.e., when the relay is close to the source, $d_{01} \rightarrow 0$; (ii) $C_5 \rightarrow 0$, i.e., when the relay is close to the destination, $d_{12} \rightarrow 0$; or (iii) $C_4 \rightarrow 0$, i.e., when the relay and the destination are about the same distance from the source, $d_{01}^{-2} - d_{02}^{-2} \rightarrow 0$ [by the triangle inequality, we have $d_{12}^{-1} \geq \frac{1}{ d_{01} + d_{02}}$].

\section{Remarks}
In this letter, we have shown that there is little to be gained by doing partial decoding rather than complete decoding at the relay for the half-duplex AWGN single-relay channel. We have analytically shown that using partial-decode-forward (PDF) can increase the rate over complete-decode-forward (CDF) by at most one-eighth.  We have also identified four scenarios in which the CDF rates asymptotically approach the PDF rates.

While a multiplicative-type bound on the gain of PDF over CDF has been obtained in this paper, we note that the absolute gain {\em might} be unbounded in the high SNR regime (we have shown in Remark~\ref{remark:large-g} that our upper bound on the absolute gain is unbounded). Hence, one may attempt PDF in the high-SNR regime when the absolute gain in transmission rate (as opposed to the percentage gain) outweighs the system complexity.

The results in this letter provide a strong motivation to extend the comparison between these two strategies to the multiple-relay channel to investigate if partial-decode-forward is beneficial when there are more relays in the network, in which partial-decode-forward is much more computationally complex than complete-decode-forward.

%\bibliography{../bib}


\end{document}